\documentclass[11pt]{article}
\usepackage{pdfpages}
\usepackage{multirow}
\usepackage{array}
\usepackage{graphicx}
\usepackage{subcaption}
\usepackage{latexsym}
\usepackage{pstricks}
\usepackage{pst-node}
\usepackage{pst-tree}
\usepackage{url}
\usepackage{times}
\usepackage{helvet}
\usepackage{courier}
\usepackage{algorithmic}
\usepackage[ruled,vlined,linesnumbered]{algorithm2e}
\usepackage{amsthm}
\usepackage{graphicx,amssymb,amsmath}
\usepackage{mathrsfs}
\usepackage{mathtools}
\usepackage{comment}

\usepackage[hypertexnames=false, colorlinks=true, citecolor=blue, linkcolor=red, urlcolor=black]{hyperref}

\usepackage{geometry}
\usepackage{authblk}

\makeatletter
\renewcommand*{\verbatim@font}{\sffamily}

\DeclareMathOperator*{\E}{\mathbb{E}}

\title{Simpler Partial Derandomization of PPSZ for $k$-SAT}

\author{S. Cliff Liu \\
Princeton University	\\
\textsf{sixuel@cs.princeton.edu}
}

\begin{document}

\maketitle


\newcounter{dummy} \numberwithin{dummy}{section}
\newtheorem{lemma}[dummy]{Lemma}
\newtheorem{definition}[dummy]{Definition}
\newtheorem{remark}[dummy]{Remark}
\newtheorem{corollary}[dummy]{Corollary}
\newtheorem{claim}[dummy]{Claim}
\newtheorem{observation}[dummy]{Observation}

\newtheorem{theorem}{Theorem}

\renewcommand{\algorithmicrequire}{\textbf{Input:}}
\renewcommand{\algorithmicensure}{\textbf{Output:}}

\begin{abstract}
    We give a simpler derandomization of the best known $k$-SAT algorithm PPSZ [FOCS'97, JACM'05] for $k$-SAT with \emph{sub-exponential} number of solutions. The existing derandomization uses a complicated construction of small sample space, while we only use \emph{hashing}.
    Our algorithm and theorem also have a nice byproduct: It outperforms the current fastest deterministic $k$-SAT algorithm when the formula has \emph{moderately exponential} number of solutions.
\end{abstract}


\section{Introduction}\label{intro}

As one of the prototypical \textsf{NP}-complete problems, the $k$-SAT problem is to decide whether a given $k$-CNF has a solution and output one if it has.
Among many revolutionary algorithms for solving $k$-SAT, two of them stick out:
Sch{\"{o}}ning's algorithm based on random walk \cite{schoning1999probabilistic} and PPSZ based on resolution  \cite{DBLP:journals/jacm/PaturiPSZ05}, where PPSZ also has profound implications in many aspects of complexity theory and algorithm design \cite{DBLP:conf/focs/PaturiPZ97, DBLP:journals/jcss/ImpagliazzoP01, DBLP:journals/jcss/ImpagliazzoPZ01, DBLP:conf/coco/CalabroIP06, DBLP:journals/siamcomp/Williams13, DBLP:conf/focs/AbboudW14}.

Both Sch{\"{o}}ning's algorithm and PPSZ are randomized:
Each try of the polynomial-time algorithms finds a solution with probability $c^{-n}$ for some $c \in (1, 2)$ where $n$ is the number of variables in the input formula.
Continuous progresses have been made in derandomizing Sch{\"{o}}ning's algorithm:
Starting with a partial derandomization in \cite{dantsin2002deterministic}, it is fully derandomized in \cite{moser2011full} and further improved in \cite{DBLP:conf/icalp/Liu18}, giving a deterministic algorithm for $3$-SAT  that runs in time $1.328^n$, which is currently the best.
In contrast, derandomizing PPSZ is notoriously hard.

It should be mentioned that the behavior of PPSZ was not completely understood at the first time it was invented: There is an exponential loss in the upper bound for General $k$-SAT, comparing with that for Unique $k$-SAT (the formula guarantees to have at most one solution).
Nevertheless, Unique $k$-SAT is believed to be at least as hard as General $k$-SAT \cite{DBLP:journals/jcss/CalabroIKP08}.
In \cite{hertli20143}, it is shown that the bound for Unique $k$-SAT holds in general, making PPSZ the current fastest randomized $k$-SAT algorithm: $3$-SAT can be solved in time $1.308^n$ with one-sided error.

The solely known result towards derandomizing PPSZ is from \cite{DBLP:conf/sat/Rolf05}, which only works for Unique $k$-SAT.
Their method of \emph{small sample space} (cf. \S{16.2} in \cite{alon2016probabilistic}) approximates the uniform distribution using a discrete subset with polynomial size, which complicates the analysis by introducing the precision of real numbers and convergence rate.
As mentioned above, the analysis of PPSZ for the General case can be much more challenging, and it is an important open question of whether this case can be derandomized even with a moderate sacrifice in the running time.

In this paper, we provide a very simple deterministic algorithm that matches the upper bound of the randomized PPSZ algorithm when the formula has sub-exponential number of solutions.
(The algorithm needs not to know the number of solutions in advance.)
Our analysis is simpler than the original randomized version \cite{DBLP:journals/jacm/PaturiPSZ05}.
Comparing with the complicated construction of small sample space in the previous derandomization \cite{DBLP:conf/sat/Rolf05}, our proof only uses hashing.

\subsection{Techniques and Main Result}

To get a sense of how our approach works, we now sketch the PPSZ algorithm and give the high-level ideas in our derandomization, then formally state our main result.

The input formula is preprocessed until no new clause of length at most $\tau$ can be obtained by pairwise resolution.
(Think of $\tau$ as a large enough integer for now.)
Each try of the algorithm processes the variables one at a time in a uniform random order:
If the variable appears in a clause of length one then it is \emph{forced} to take the only truth value to satisfy the clause, otherwise it is \emph{guessed} to take a uniform random truth value.
The probability of finding a solution is decided by the number of guessed variables, which contain two parts: the \emph{frozen} variables that take the same truth values in all solutions, and the others called the \emph{liquid} variables.
There are two places in PPSZ that use randomness:
\begin{enumerate}
    \item The random order of the variables. \label{rd_1}
    \item The random values assigned to the guessed variables. \label{rd_2}
\end{enumerate}
To remove the randomness in (\ref{rd_1}), the key observation is that in the Unique case, one only needs the order of every $\tau$ variables to be uniformly random, which can be achieved by using a $\tau$-wise independent distribution, then the (frozen) variables are forced with roughly the same probability as that using mutually independence.
By choosing $\tau$ wisely, such a distribution exists with sub-exponential support.
To remove the randomness in (\ref{rd_2}), we enumerate all possible truth values of the guessed variables.
To obtain a good time bound, one needs to bound the number of guessed variables.
In the Unique case, all variables are frozen in the original input formula.
In the General case, we show that there exists a \emph{good} order of variables such that the number of liquid variables is upper bounded by some function of the number of solutions when the variables are fixed to certain values according to this order.
After fixing the liquid variables, we reduce to the Unique case.
The expected number of frozen variables that are guessed can be upper bounded using the \emph{frozen tree}, which is simplified from the \emph{critical clause tree} in \cite{DBLP:journals/jacm/PaturiPSZ05}.

Our derandomization is \emph{partial} in the sense that the randomness in (\ref{rd_1}) and (\ref{rd_2}) is completely removed with sub-exponential slowdown, but the reduction from General to Unique introduces an additional factor in the running time depending on the number of solutions.
The main result is stated below:
\begin{theorem}[Main Result]\label{thm_main}
    There exists a deterministic algorithm for $k$-SAT
    such that for any $k$-CNF $F$ on $n$ variables with $2^{\delta n + o(n)}$ solutions,
    the algorithm outputs a solution of $F$ in time
    \begin{equation*}
        2^{(1 - \lambda_k + \lambda_k \delta + \rho(\delta)) n + o(n)},
    \end{equation*}
    where $\rho(\delta) = -\delta \log_2 \delta - (1-\delta) \log_2(1-\delta)$ is the \emph{binary entropy function} and
    \footnote{The values of $\rho(0)$ and $\rho(1)$ are defined to be $0$. Some typical values: $\lambda_3 = 2 - 2 \ln 2 \approx 0.6137$ and $\lambda_4 \approx 0.4452$.} 
    \begin{equation*}
        \lambda_k = \sum_{j=1}^{\infty} \frac{1}{j (kj - j + 1)}.
    \end{equation*}
\end{theorem}

This matches the upper bound $2^{(1 - \lambda_k)n + o(n)}$ of the randomized PPSZ algorithm when the formula has $2^{o(n)}$ solutions.
Our algorithm and theorem also have a nice byproduct:
It is faster than the current best deterministic $k$-SAT algorithm \cite{DBLP:conf/icalp/Liu18} when the formula has \emph{moderately exponential} number of solutions.
For example, for $3$-SAT with at most $2^{n / 480}$ solutions and $4$-SAT with at most $2^{n/361}$ solutions, our deterministic algorithm is currently the fastest.

\section{Preliminaries}\label{pre}

We begin with some basic notations and definitions,
then we review the PPSZ algorithm and formalize it under our framework.

\subsection{Notations}

The formula is in Conjunctive Normal Form (CNF).
Let $V$ be a finite set of Boolean \emph{variables} each taking value from $\{0, 1\}$.
A \emph{literal} $l$ over $x \in V$ is either $x$ or $\bar{x}$, and $V(l)$ is used to denote the variable $x$ corresponding to $l$.
A \emph{clause} $C$ over $V$ is a finite set of literals over distinct variables from $V$. We use $V(C)$ to denote the set of all variables in $C$.
A \emph{formula} $F$ is a finite set of clauses,
and is a \emph{$k$-CNF} if every clause in $F$ contains at most $k$ literals.
We use $V(F)$ to denote the set of all variables in $F$.
If the context is clear, we omit $F$ and only use $V$ to denote $V(F)$.

An \emph{assignment} $\alpha$ is a finite set of literals over distinct variables.
Let $V(\alpha)$ be $\{V(l) \mid l \in \alpha\}$,
then $\alpha$ is called a \emph{complete assignment} of formula $F$ if $V(F) \subseteq V(\alpha)$ and is called a \emph{partial assignment} (usually denoted by $a$ to distinguish from $\alpha$) of $F$ if otherwise,
and we call them just an \emph{assignment} if the context is well-understood.
A literal $x$ (resp. $\bar{x}$) is \emph{satisfied} by $\alpha$ if $x \in \alpha$ (resp. $\bar{x} \in \alpha$).
A clause $C$ is \emph{satisfied} by $\alpha$ if $C$ contains a literal satisfied by $\alpha$, otherwise $C$ is \emph{falsified} by $\alpha$ if additionally $V(C) \subseteq V(\alpha)$.
A formula $F$ is \emph{falsified} by $\alpha$ if there is a clause in $F$ falsified by $\alpha$.
Otherwise, $F$ is \emph{satisfied} by $\alpha$ if $\alpha$ satisfies all the clauses of $F$, and if additionally $V(F) = V(\alpha)$ then we call $\alpha$ a \emph{solution} of $F$.

We use $\text{sat}(F)$ to denote the set of all the solutions of $F$ and let $S(F) \coloneqq |\text{sat}(F)|$.
We also use $S$ to denote the number of solutions of the original input formula.
The \emph{$k$-SAT} problem is to decide whether a given $k$-CNF $F$ is \emph{satisfiable} (having at least one solution) and output one if it has.

Given a CNF $F$ and a literal $l$, we use $F_l$ to represent the CNF by
deleting the literal $\bar{l}$ in all clauses of $F$ and deleting the clauses of $F$ that contain $l$.
In general, we use $F_{\alpha}$ to denote the CNF by doing the above on $F$ for all the literals in (partial) assignment $\alpha$.
Obviously, if this creates an empty clause (denoted by $\bot$) then $\alpha$ falsifies this clause and also falsifies $F$, and if this leaves no clause in $F$ then $\alpha$ satisfies $F$.

We will assume in the rest of the paper that the input formula $F$ is a satisfiable $k$-CNF with $k \ge 3$ a fixed integer, then $\lambda$ is used to denote the $\lambda_k$ defined in Theorem~\ref{thm_main}.
Let $n$ be the number of variables in the input formula $F$, we assume that $F$ has $|F| = \text{poly}(n)$ clauses.
We use $\log$ to denote the base-two logarithm,
and use $\widetilde{O}(f(n)) = 2^{o(n)} \cdot f(n)$ to suppress sub-exponential factors. 

\subsection{The PPSZ Algorithm}\label{subsec_ppsz}

In this section, we review the (randomized) PPSZ algorithm, the key definitions, and some previous results under our framework, for the purpose of derandomization.

A modified PPSZ algorithm is presented (Algorithm~\ref{alg_PPSZ} and Algorithm~\ref{alg_modify}), which is slightly different from the original version \cite{DBLP:journals/jacm/PaturiPSZ05} as well as its variants \cite{hertli20143, DBLP:conf/coco/SchederS17}.
The algorithm relies on the following concept:

\begin{definition}[\cite{hertli20143}]\label{def_tau_imply}
    Let $F$ be a CNF, a literal $l$ is \emph{implied} by $F$ if $l \in \bigcap_{\alpha \in \text{sat}(F)} \alpha$.
    Let $\tau$ be a positive integer,
    a literal $l$ is \emph{$\tau$-implied} by $F$ if there exists a CNF $J \subseteq F$ with $|J| \le \tau$ such that $J$ implies $l$.
\end{definition}

The PPSZ algorithm outlined in Algorithm~\ref{alg_PPSZ} is randomized, but its subroutine \textsf{Modify} (Algorithm~\ref{alg_modify}) is deterministic as long as its line~\ref{line_find_implied} is deterministic, which we now specify:

\begin{remark}\label{imply_to_resolution}
    To find a $\tau$-implied literal $l$ in $F_a$ with $V(l) = x$,
    for each $J \subseteq F_a$ with $|J| \le \tau$, compute $\text{sat}(J)$ and check whether the intersection contains $l$,
    which can be done deterministically in time $O({|F|}^{\tau} \cdot 2^{k \tau}) = n^{O(\tau)}$.
    \footnote{The time bound here is obtained by a naive enumeration. This is different from the method of \emph{bounded resolution} in \cite{DBLP:journals/jacm/PaturiPSZ05}. The name \textsf{Modify} also comes from there.}
\end{remark}

\begin{algorithm}[h]
\caption{$\textsf{PPSZ}(F, \Sigma, \tau)$}
\label{alg_PPSZ}
\begin{algorithmic}[1]
\REQUIRE $k$-CNF $F$, set $\Sigma$ of permutations on $V$, integer $\tau$
\ENSURE $\bot$ or solution $\alpha$
\STATE choose a permutation $\sigma$ from $\Sigma$ uniformly at random \label{line_uni_per}
\STATE choose a bit vector $\beta$ from $\{0, 1\}^{n}$ uniformly at random  \label{line_uni_bit}
\RETURN $\textsf{Modify}(F, \sigma, \beta, \tau)$
\end{algorithmic}
\end{algorithm}

\begin{algorithm}[h]
\caption{$\textsf{Modify}(F, \sigma, \beta, \tau)$}
\label{alg_modify}
\begin{algorithmic}[1]
\REQUIRE $k$-CNF $F$, permutation $\sigma$, bit vector $\beta$, integer $\tau$
\ENSURE $\bot$ or solution $\alpha$
\STATE initialize assignment $a$ as an empty set
\FOR {each $x \in V$ in the order of $\sigma$} \label{line_loop_begin}
    \IF {all bits in $\beta$ have been exhausted}{
        \RETURN $\bot$
    }
    \ENDIF
    \IF {$F_{a}$ contains a $\tau$-implied literal $l$ with $V(l) = x$}{ \label{line_find_implied}
        \STATE add $l$ to $a$ \label{line_forced}
    }
    \ELSE {
        \STATE set $l$ to $x$ if the next bit of $\beta$ is $1$ and to $\bar{x}$ if otherwise, add $l$ to $a$ \label{line_guessed}
    }
    \ENDIF
\ENDFOR \label{line_loop_end}
\STATE \textbf{if} $a$ satisfies $F$ \textbf{then} \textbf{return} $a$ as $\alpha$, otherwise \textbf{return} $\bot$
\end{algorithmic}
\end{algorithm}

With the algorithms well defined, we can formalize the \emph{success probability} of PPSZ:
\begin{definition}\label{def_U}
    Given $k$-CNF $F$, a set $\Sigma$ of permutations on $V$, and integer $\tau$, define
    \begin{equation*}
        \Pr[\text{Success}] \coloneqq \Pr[\textsf{PPSZ}(F, \Sigma, \tau) \in \text{sat}(F)] = \Pr_{\sigma \sim U_{\Sigma}, \beta \sim U_n}[\textsf{Modify}(F, \sigma, \beta, \tau) \in \text{sat}(F)],
    \end{equation*}
    where probability distributions $U_{\Sigma}: \Sigma \mapsto [0, 1]$, $U_n: \{0,1\}^n \mapsto [0,1]$ are uniform distributions.
\end{definition}
With Definition~\ref{def_U},
we are now ready to state the main result for the randomized PPSZ algorithm:
\begin{theorem}[\cite{DBLP:journals/jacm/PaturiPSZ05, hertli20143}]\label{thm_old_ppsz}
    If $\Sigma = \text{Sym}(V)$ and $\tau = \log n$, then $\Pr[\text{Success}] \ge 2^{-(1 - \lambda) n - o(n)}$.
\end{theorem}
In the rest of the paper, fix $\tau = \log n$ and omit the parameter $\tau$ in the algorithms.

By Remark~\ref{imply_to_resolution}, $\textsf{Modify}(F, \sigma, \beta)$ runs in time $O(n \cdot |F| \cdot n^{O(\log n)}) = \widetilde{O}(1)$.
Therefore, by Theorem~\ref{thm_old_ppsz} and a routine argument, we obtain a randomized algorithm for $k$-SAT with one-sided error, whose upper bound of the running time is $\widetilde{O}(2^{(1 - \lambda) n})$.

To showcase the proof of Theorem~\ref{thm_old_ppsz}, and more importantly, to motivate our derandomization, we need the following key definitions:
\begin{definition}\label{def_step_alpha_x}
    We call each iteration of the loop (lines~\ref{line_loop_begin}-\ref{line_loop_end}) in $\textsf{Modify}(F, \sigma, \beta)$ a \emph{step}.
    For any variable $x \in V$, let $a(x)$ be the partial assignment  $a$ at the beginning of step $i$ where $\sigma(i) = x$.
\end{definition}
\begin{definition}[\cite{hertli20143, DBLP:conf/coco/SchederS17}]\label{def_frozen_liquid_forced_guessed}
    A literal $l$ is \emph{frozen} in $F$ if $l$ is implied by $F$.
    A variable $x$ is \emph{frozen} in $F$ if the literal $x$ or $\bar{x}$ is frozen in $F$, otherwise $x$ is \emph{liquid} in $F$. \footnote{In the literature, frozen variables are also called \emph{critical variables} or \emph{backbones} \cite{biere2009handbook}.}

    A literal $l$ is \emph{forced} if $l$ is $\tau$-implied by $F_{a(x)}$ such that $V(l) = x$.
    A variable $x$ is \emph{forced} if the literal $x$ or $\bar{x}$ is forced, otherwise $x$ is \emph{guessed}.
\end{definition}


\begin{definition}\label{def_LG}
    Given any assignment $\alpha \in \text{sat}(F)$ and variable $x \in V$, in the execution of $\textsf{Modify}(F, \sigma, \beta)$ that returns $\alpha$, indicator $G_x(\alpha, \sigma)$ is $1$ if and only if $x$ is guessed,
    and let $G(\alpha, \sigma) \coloneqq \sum_{x \in V} G_x(\alpha, \sigma)$.
\end{definition}

By an induction on steps, $\textsf{Modify}(F, \sigma, \beta)$ returns $\alpha$ if and only if all the guessed variables are \emph{correctly guessed} to take the corresponding literals in $\alpha$.
Thus by Definition~\ref{def_U}, it is easy to see that
\begin{equation}\label{ppsz_prob}
    \Pr[\text{Success}] =  \sum_{\alpha \in \text{sat}(F)} \E_{\sigma \sim U_{\Sigma}} \left[ 2^{-G(\alpha, \sigma)} \right] .
\end{equation}

In \cite{DBLP:conf/coco/SchederS17}, it is shown that the right-hand side of Equality~(\ref{ppsz_prob}) is lower bounded by $2^{-(1 - \lambda) n - o(n)}$ by considering two probability distributions on $\text{sat}(F) \times \Sigma$ and applying Jensen's inequality, from which Theorem~\ref{thm_old_ppsz} is immediate.

\section{Derandomization for the Unique Case}\label{sec_framework}

Our derandomization of PPSZ (Algorithm~\ref{alg_dPPSZ}) for the Unique case is simple:
Enumerating all possible lengths $i$ (from $1$ to $n$) of bit vectors, all bit vectors $\beta$ in $\{0, 1\}^i$, and all permutations $\sigma$ in $\Sigma$ to run $\textsf{Modify}(F, \sigma, \beta)$.

\begin{algorithm}[h]
\caption{$\textsf{dPPSZ}(F, \Sigma)$}
\label{alg_dPPSZ}
\begin{algorithmic}[1]
\REQUIRE $k$-CNF $F$, set $\Sigma$ of permutations on $V$
\ENSURE solution $\alpha$
\FOR {each $i$ from $1$ to $n$} { \label{line_loop1}
    \FOR {each bit vector $\beta \in \{0, 1\}^i$} { \label{line_loop2}
        \FOR {each permutation $\sigma \in \Sigma$} {\label{line_loop3}
                \IF {$\textsf{Modify}(F, \sigma, \beta) \ne \bot$} {
                    \RETURN $\textsf{Modify}(F, \sigma, \beta)$
                }
                \ENDIF
        }
        \ENDFOR
    }
    \ENDFOR
}
\ENDFOR \label{line_loop1_end}
\end{algorithmic}
\end{algorithm}

We call each iteration of the outer loop (lines~\ref{line_loop1}-\ref{line_loop1_end}) in Algorithm~\ref{alg_dPPSZ} a \emph{round}.
By Theorem~\ref{thm_old_ppsz}, if $\Sigma = \text{Sym}(V)$ then $\Pr[\text{Success}] > 0$, thus $\textsf{dPPSZ}(F, \Sigma)$ returns a solution since the last round (round $n$) must find a solution.
However, it runs in time $\widetilde{O}(n! 2^n)$, which is even worse than a brute-force search.

The goal is to construct a small $\Sigma$ such that $\textsf{dPPSZ}(F, \Sigma)$ finds a solution in at most $q$ rounds for some reasonably small $q$.
Note that if $|\Sigma| = \widetilde{O}(1)$ and
\begin{equation}\label{set_q}
    q = (1 - \lambda) n + o(n),
\end{equation}
then $\textsf{dPPSZ}(F, \Sigma)$ runs in time
\begin{equation}\label{eq_running_time}
    \widetilde{O}(\sum_{i \in [q]} 2^i \cdot |\Sigma|) = \widetilde{O}(2^q) = \widetilde{O}(2^{(1 - \lambda)n} )
\end{equation}
as desired. In the rest of this section, we shall fix $q$ as the value in the right-hand side of Equality~(\ref{set_q}).


Let $\alpha$ be the only element in $\text{sat}(F)$.
The observation is that
if there exists a $\sigma \in \Sigma$ such that $G(\alpha, \sigma) \le q$,
then enumerating all $\sigma$ and all bit vectors of length $q$
guarantees to find $\alpha$,
because the number of guessed variables, or equivalently, the number of used bits in $\beta$ in the execution of $\textsf{Modify}(F, \sigma, \beta)$ that returns $\alpha$ is at most $q$.
It remains to find such a $\Sigma$ with acceptable size.
We shall need the following definition:
\begin{definition}\label{def_enumerable}
    A permutation set $\Sigma$ is \emph{enumerable} if each permutation in $\Sigma$ is on $V$, $|\Sigma| = \widetilde{O}(1)$, and $\Sigma$ can be deterministically constructed in time $\widetilde{O}(1)$.
\end{definition}

In the rest of this section, we will prove the following main lemma (Lemma~\ref{lem_new_bits}), which implies the derandomization (Theorem~\ref{thm_unique}):
\begin{lemma}[Main Lemma]\label{lem_new_bits}
    If $F$ has exactly one solution then there exists an enumerable permutation set $\Sigma$ such that
    \begin{equation*}
        \E_{\sigma \sim U_{\Sigma}} \left[ G_x(\alpha, \sigma)\right] \le 1 - \lambda + o(1)
    \end{equation*}
    for any variable $x \in V$.
\end{lemma}

\begin{theorem}\label{thm_unique}
    There exists a deterministic algorithm for Unique $k$-SAT that runs in time $2^{(1 - \lambda)n + o(n)}$.
\end{theorem}
\begin{proof}
    By Definition~\ref{def_enumerable}, we first construct an enumerable $\Sigma$ in $\widetilde{O}(1)$ time then call $\textsf{dPPSZ}(F, \Sigma)$.
    By Lemma~\ref{lem_new_bits} and linearity of expectation, we obtain $\E_{\sigma \sim U_{\Sigma}} \left[ G(\alpha, \sigma)\right] \le q$, which means that there exists a $\sigma \in \Sigma$ such that $G(\alpha, \sigma) \le q$.
    The theorem follows from the observation in the previous discussion.
\end{proof}

We start with constructing an enumerable permutation set in \S{\ref{subsec_sigma}}, which will be the required $\Sigma$ in Lemma~\ref{lem_new_bits}.
To proceed with our proof, we introduce our vital combinatorial structure called the \emph{frozen tree} in \S{\ref{subsec_dft}}.
After that, in \S{\ref{subsec_lem_unique}} we finish the proof of Lemma~\ref{lem_new_bits}.

\subsection{A Small Permutation Set}\label{subsec_sigma}

We shall use a $K$-wise independent hash family to construct our enumerable permutation set.
First of all, we review the basic definition:
\begin{lemma}[cf. \S{3.5.5} in \cite{DBLP:journals/fttcs/Vadhan12}]\label{lem_k_wise}
    For $N, M, K \in \mathbb{N}$ such that $K \le N$, a family of functions $H = \{h: [N] \mapsto [M]\}$ is \emph{$K$-wise independent} if for all distinct $x_1, \dots, x_K \in [N]$, the random variables $h(x_1), \dots h(x_K)$ are independent and uniformly distributed in $[M]$ when $h$ is chosen from $H$ uniformly at random.
    The size of $H$ and the time to deterministically construct $H$ can be
    $\text{poly}((\max\{M, N\})^K)$.
\end{lemma}

We construct the permutation set $\Sigma$ by $\textsf{Construct-$\Sigma$}(V)$ (Algorithm~\ref{alg_construct_sigma}).
The parameters in the algorithm are set with anticipation of what we will do in the later analysis.
\begin{algorithm}[h]
\caption{$\textsf{Construct-}\Sigma(V)$}
\label{alg_construct_sigma}
\begin{algorithmic}[1]
\REQUIRE variable set $V$
\ENSURE permutation set $\Sigma$
\STATE initialize $\Sigma$ as an empty set, $n$ as $|V|$, and $\tau$ as $\log n$
\STATE let $N \coloneqq n$, $M \coloneqq n$, and $K \coloneqq \tau$, then construct hash family $H$ as in Lemma~\ref{lem_k_wise} \label{set_alg4_par}
\STATE assign each $x \in V$ a distinct index $i(x) \in [n]$
\FOR {each $h \in H$} {
    \FOR {each $x \in V$} {\label{line_gamma_value_begin}
        \STATE set $\gamma(x) \coloneqq h \circ i(x)$
    }
    \ENDFOR\label{line_gamma_value_end}
    \STATE sort all $x \in V$ according to $\gamma(x)$ in ascending order, breaking ties by an arbitrary deterministic rule \label{line_sorting}
    \STATE let the sorted order of variables be $\sigma$ and add $\sigma$ to $\Sigma$
}
\ENDFOR
\RETURN $\Sigma$
\end{algorithmic}
\end{algorithm}

Lines~\ref{line_gamma_value_begin}-\ref{line_gamma_value_end} of Algorithm~\ref{alg_construct_sigma} define a function $\gamma: V \mapsto [n]$.  For any $x \in V$, $\gamma(x)$ is called the \emph{placement} of $x$.
Note that $\textsf{Construct-$\Sigma$}(V)$ returns a multiset $\Sigma$ since $|H| \ge n^{\tau}$ is greater than the number of all possible orders ${\tau}! \cdot \binom{n}{\tau}$, however we keep all the duplicates in $\Sigma$ for the following reason:
\begin{remark}\label{hash_to_permutation}
    $\textsf{Construct-$\Sigma$}(V)$ defines a bijection between $H$ and $\Sigma$, therefore choosing a $\sigma \in \Sigma$ uniformly at random is equivalent to choosing an $h \in H$ uniformly at random.
    For a $\sigma \in \Sigma$ chosen uniformly at random and all distinct variables $x_1, \dots, x_{\tau} \in V$, the random variables $\gamma(x_1), \dots, \gamma(x_{\tau})$ are independent and uniformly distributed in $[n]$.
\end{remark}

By line~\ref{set_alg4_par} of Algorithm~\ref{alg_construct_sigma}, Lemma~\ref{lem_k_wise}, and Remark~\ref{hash_to_permutation}, the permutation set returned by $\textsf{Construct-$\Sigma$}(V)$ is enumerable.

\subsection{The Frozen Tree}\label{subsec_dft}

Our frozen tree is different from the so-called \emph{critical clause tree} in \cite{DBLP:journals/jacm/PaturiPSZ05} for two main reasons.
Firstly, we are using $\tau$-implication in each step rather than bounded resolution as a preprocessing.
Secondly, we do not need the tree for General $k$-SAT.

Recall that a subset $A$ of vertices in a rooted tree is a \emph{cut} if it does not include the root and every path from the root to a leaf contains exactly one vertex in $A$.
For convenience, we introduce a dummy variable $\kappa \notin V$ and let $\alpha$ be the union of the literal set $\{\kappa\}$ with a solution of $F$.
\begin{definition}\label{def_frozen_tree}
    Given a positive integer $K$ and a variable $x \in V$, a rooted tree is a \emph{$K$-frozen tree} for $x$ if it has the following properties:
    \begin{enumerate}
        \item The root $u$ is labeled by $x$, and any other vertex is labeled by $\kappa$ or a variable in $V$.
            If $v$ is labeled by $y$, let $l(v)$ be the literal $l \in \alpha$ with $V(l) = y$.
            For any subset $W$ of vertices in the tree, let $V(W)$ be the set of the labels of all vertices in $W$. \label{pp_r}
        \item Any vertex has at most $k-1$ children. \label{pp_k}
        \item All vertices on the path $P(v)$ from the root to any vertex $v$ with labels different from $\kappa$ have distinct labels. \label{pp_p}
        \item Any leaf is at depth $d = \lfloor \log_k K \rfloor$. \label{pp_d}
        \item The number of different labels except $\kappa$ in the tree is at most $K$. \label{pp_l}
        \item For any cut $A$, let $\alpha(A) \coloneqq \{l(v) \mid v \in A \}$,
            then $l(u)$ is $K$-implied by $F_{\alpha(A)}$. \label{pp_c}
    \end{enumerate}
\end{definition}

In the rest of this section, we shall prove the following lemma:
\begin{lemma}\label{lem_fv_ft}
    For any frozen variable $x$ there exists a $\tau$-frozen tree for $x$.
\end{lemma}

As observed in \cite{hertli20143}, it is possible to extend the proof for the existence of a critical clause tree in \cite{DBLP:journals/jacm/PaturiPSZ05} to obtain a proof for Lemma~\ref{lem_fv_ft}.
Our proof here is simpler and self-contained, which might be of independent interest.

\begin{algorithm}[h]
\caption{\textsf{Construct-Tree}$(F, a, v, y, d)$}
\label{alg_construct_tree}
\begin{algorithmic}[1]
\REQUIRE $k$-CNF $F$, assignment $a$, tree vertex $v$, variable $y \in V$, integer $d$
\ENSURE a tree rooted at $v$
\STATE label $v$ by $y$ and replace $l$ in $a$ with $\bar{l}$ where $V(l) = y$ \label{line_update_a}
\IF {$d > 0$} {
    \STATE choose a clause $C(v) \in F$ falsified under $a$ \label{line_choose_clause}
    \FOR {each variable $y' \in V(C(v))$ such that $y' \notin V(P(v))$} { \label{line_create_child_w}
        \STATE create a child vertex $v'$ of $v$ and \textsf{Construct-Tree}$(F, a, v', y', d -1)$
    }
    \ENDFOR
    \IF {no child of $v$ is created} {\label{line_add_chi}
        \STATE create a child vertex $v'$ of $v$ and \textsf{Construct-Tree}$(F, a, v', \kappa, d -1)$
    }
    \ENDIF
}
\ENDIF
\RETURN the tree rooted at $v$
\end{algorithmic}
\end{algorithm}

    We call \textsf{Construct-Tree}$(F, \alpha, u, x, \lfloor \log_k \tau \rfloor)$ (Algorithm~\ref{alg_construct_tree}) to construct a tree rooted at $u$
    and prove that the constructed tree is a $\tau$-frozen tree for $x$ by proving all six properties in Definition~\ref{def_frozen_tree} true.

    We will frequently use the clause $C(v)$ associated to any vertex $v$ in the tree, which is guaranteed to exist since $a$ contains the wrong literal over $x$ (line~\ref{line_choose_clause}), which must be frozen since $F$ has exactly one solution.
    Properties~\ref{pp_r} and \ref{pp_d} trivially hold.
    Property~\ref{pp_k} holds for each vertex $v$ since the clause $C(v)$ has size at most $k$ and must be satisfied if no variable of it appears in $V(P(v))$, thus the loop in line~\ref{line_create_child_w} runs for at most $k-1$ times.
    Property~\ref{pp_p} follows directly from line~\ref{line_create_child_w}.
    By Properties~\ref{pp_k} and \ref{pp_d}, the number of vertices in the constructed tree is at most $\sum_{i=0}^d (k-1)^i \le k^d \le \tau$, then Property~\ref{pp_l} follows immediately.

    It remains to prove Property~\ref{pp_c} true. We need the following lemma:
    \begin{lemma}\label{claim_ancestor}
        Given a subtree rooted at a non-leaf vertex $v$ and $A$ the set of all children of $v$,
        for any assignment $\alpha_A \in \text{sat}(J_{\alpha(A)})$ where CNF $J = \{ C(w') \mid w' \in P(v) \}$,
        there exists an ancestor $w$ of $v$ such that $l(w) \in \alpha_A$.
        \footnote{Conventionally, a vertex is also considered an ancestor of itself. }
    \end{lemma}
    \begin{proof}
    Observe that $C(v) \in J_{\alpha(A) \cup \beta}$ is falsified, where $\beta$ is the set of all literals $\bar{l}$ in line~\ref{line_update_a} of the executions of \textsf{Construct-Tree}$(F, a, w', y, d)$ for all $w' \in P(v)$.
    Indeed, all the variables in clause $C(v)$ are from $V(A)$ and $V(P(v))$, whose corresponding literals are all falsified in $\alpha(A) \cup \beta$ by line~\ref{line_choose_clause}. So at least one literal in $\beta$ is not in $\alpha_A$, giving the lemma.
    \end{proof}

    To prove Property~\ref{pp_c}, we need to identify a CNF $J$ consisting of at most $\tau$ clauses such that $J$ implies $l(u)$.
    Let $T$ be the set of all vertices in the tree,
    we claim that the desired $J$ can be $\{ C(v) \mid v \in V(T) \}$, which consists of at most $\tau$ clauses since there are at most $\tau$ vertices in $T$.

    The remaining proof is by an induction on $d$.
    If $d = 1$ then the cut is the set of all children of the root $u$.
    Thus by Lemma~\ref{claim_ancestor}, the only ancestor $u$ suffices that $l(u)$ is in any solution of $J_{\alpha(A)}$, giving the lemma.
    Now suppose the lemma holds for $d = i$ and we prove it for $d = i + 1$.
    We shall use the following lemma:
    \begin{lemma}\label{claim_cut}
        Given a cut $A$, for any $\alpha' \in \text{sat}(J_{\alpha(A)})$, there exists a vertex set $\widetilde{A} \subseteq T$, such that $\alpha(\widetilde{A}) \subseteq \alpha' \cup \alpha(A)$ and any vertex in $\widetilde{A}$ has depth at most $i$. Furthermore, $\widetilde{A}$ is either a cut or contains $u$.
    \end{lemma}
    \begin{proof}
    We shall construct $\widetilde{A}$ by the following process.
    Initialize $\widetilde{A}$ as $A$, we repeatedly modify $\widetilde{A}$ until it contains $u$ or any vertex in it has depth at most $i$ while keeping $\widetilde{A}$ a cut.
    Choose a vertex $v'$ in $\widetilde{A}$ at depth $i+1$, let $v$ be its parent and let $A'$ be the set of all children of $v$.
    Since no ancestor of $v$ is in $\widetilde{A}$ by the definition of a cut, it must be that  $A' \subseteq \widetilde{A}$.
    Since $\alpha'$ satisfies $J_{\alpha(A)}$, we have that $\alpha'$ also satisfies its subset $J'_{\alpha(A)}$ where $J' = \{ C(w') \mid w' \in P(v) \}$.
    Note that $\alpha(A) = \alpha(A') \cup \alpha(A \backslash A')$, thus $\alpha' \cup \alpha(A \backslash A')$ satisfies $J'_{\alpha(A')}$.
    By Lemma~\ref{claim_ancestor},
    there exists an ancestor $w$ of $v$ such that $l(w) \in \alpha' \cup \alpha(A \backslash A')$.
    If $w = u$ then we stop.
    Otherwise we replace all vertices with ancestor $w$ in $\widetilde{A}$ by $w$ to keep $\widetilde{A}$ a cut.
    Continue this process until there is no vertex in $\widetilde{A}$ at depth $i+1$.
    After the process, any vertex in the resulting $\widetilde{A}$ has depth at most $i$ and $\widetilde{A}$ is a cut if it does not contain $u$.
    Furthermore, any literal in $\alpha$ over a label from $\widetilde{A} \backslash A$ is in $\alpha' \cup \alpha(A \backslash A')$ for some $A' \subseteq A$, thus also in $\alpha' \cup \alpha(A)$. We conclude that $\alpha(\widetilde{A}) \subseteq \alpha' \cup \alpha(A)$, giving the lemma.
    \end{proof}

    By $u \notin A$ and Property~\ref{pp_p}, the label $x$ of $u$ does not appear in $A$, so either $l(u)$ or $\bar{l}(u)$ is in $\alpha'$ since $\alpha' \cup \alpha(A)$ is a solution of $J$ in which $x$ appears.
    If $\widetilde{A}$ contains $u$, then $l(u) \in \alpha'$ by Lemma~\ref{claim_cut}.
    Otherwise, we ignore all vertices below depth $i$ to obtain a tree with uniform depth $i$ and a cut $\widetilde{A}$ for it.
    Assume for contradiction that $\bar{l}(u) \in \alpha'$, then by Lemma~\ref{claim_cut} we have
    \begin{equation}\label{ineq_alpha}
        \alpha(\widetilde{A}) \cup \{\bar{l}(u)\} \subseteq \alpha' \cup \alpha(A) \cup \{\bar{l}(u)\} = \alpha' \cup \alpha(A).
    \end{equation}

    Since $\alpha' \in \text{sat}(J_{\alpha(A)})$, $J_{\alpha' \cup \alpha(A)}$ must be satisfiable.
    Thus by (\ref{ineq_alpha}), $J_{\alpha(\widetilde{A}) \cup \{\bar{l}(u)\} }$ is also satisfiable,
    contradicting with the induction hypothesis that any solution of $J_{\alpha(\widetilde{A})}$ contains $l(u)$.
    So it must be that $l(u) \in \alpha'$.
    Therefore, any solution of $J_{\alpha(A)}$ contains $l(u)$ and thus Property~\ref{pp_c} holds, completing the proof of Lemma~\ref{lem_fv_ft}.

\subsection{Proof of the Main Lemma}\label{subsec_lem_unique}

In this section, we prove Lemma~\ref{lem_new_bits}.
First of all we relate the event
$G_x(\alpha, \sigma) = 0$ for a (frozen) variable $x$, or equivalently, the event that $x$ is forced in $F_{a(x)}$, to an event in another probability space.
By Lemma~\ref{lem_fv_ft}, there exists a $\tau$-frozen tree for $x$.
So by Property~\ref{pp_c} in Definition~\ref{def_frozen_tree} and the fact that $\kappa$ does not appear in $F$,
if there exists a cut $A$ in this tree such that
all labeling variables of $A$ except $\kappa$ \emph{appear before} $x$ in the permutation $\sigma$ (denote this event by $B(\sigma)$),
then $x$ is $\tau$-implied by $F_{\alpha(A)}$,
which means that $x$ must be forced in $F_{a(x)}$ since $\alpha(A) \subseteq a(x)$ (cf. Definition~\ref{def_step_alpha_x}).
Thus
\begin{equation*}
    \Pr_{\sigma \sim U_{\Sigma}}[B(\sigma)] \le \Pr_{\sigma \sim U_{\Sigma}}[G_x(\alpha, \sigma) = 0] = 1 - \E_{\sigma \sim U_{\Sigma}}[G_x(\alpha, \sigma)] .
\end{equation*}
Therefore to prove Lemma~\ref{lem_new_bits}, it suffices to prove the following (for readability, omit the parameter $\sigma$ in $B$ and random variable $\sigma \sim U_{\Sigma}$ throughout this section):
\begin{equation}\label{ineq_B}
    \Pr[B] \ge \lambda - o(1).
\end{equation}

Recall from \S{\ref{subsec_sigma}} that the permutation $\sigma$ on $V$ is decided by the placement function $\gamma$: Variables are sorted in ascending order according to their placements (breaking ties arbitrarily, line~\ref{line_sorting} of Algorithm~\ref{alg_construct_sigma}).
Let $\widehat{B}$ be the event that there exists a cut $A$ in the tree such that all labeling variables of $A$ except $\kappa$ have \emph{strictly smaller} placements than $x$, then
\begin{equation}\label{ineq_bb}
    \Pr[B] \ge \Pr[\widehat{B}].
\end{equation}

With foresight for simplicity in the later analysis, we shall consider the following events:
\begin{definition}
    Given a $\tau$-frozen tree for $x$, for any integer $j \in [0, d]$, let $T_j$ be a subtree rooted at a vertex at depth $d-j$ and labeled $y \neq \kappa$, and let $\widetilde{B}_j$ be the event that there exists a cut $A$ in $T_j$ such that all labeling variables of $A$ except $\kappa$ have placements \emph{at most} $\gamma(y)$.
\end{definition}
\begin{lemma}\label{lem_B_Phi}
    $\Pr[\widehat{B}] \ge \Pr[\widetilde{B}_d] - o(1)$.
\end{lemma}
\begin{proof}
    $\widetilde{B}_d$ is the event defined for tree $T_d(x)$, which is the $\tau$-frozen tree.
    By Property~\ref{pp_p} in Definition~\ref{def_frozen_tree}, any label $y$ of the vertex below the root is different from $x$.
    Thus by Remark~\ref{hash_to_permutation}, $\Pr[\gamma(x) = \gamma(y)] = 1/n$.
    By a union bound over all the labels except $\kappa$ (whose number is at most $\tau$ by Property~\ref{pp_l} in Definition~\ref{def_frozen_tree}), with probability at most $\tau / n = \log n / n = o(1)$ there exists a variable with the same placement with $x$.
    Finally, we obtain $\Pr[\widetilde{B}_d] \le \Pr[\widehat{B}] + o(1)$ by inspecting the events.
\end{proof}

By Lemma~\ref{lem_B_Phi} and Inequality~(\ref{ineq_bb}), to prove Inequality~(\ref{ineq_B}), it suffices to prove the following:
\begin{equation}\label{ineq_B_d}
    \Pr[\widetilde{B}_d] \ge \lambda - o(1),
\end{equation}
which gives Lemma~\ref{lem_new_bits} by the discussion in the first paragraph of \S{\ref{subsec_lem_unique}}.

By Remark~\ref{hash_to_permutation}, for any integer $j \in [0, d]$, we write
\begin{equation}\label{eq_total_prob}
    \Pr[\widetilde{B}_j] = \sum_{r \in [n]} \Pr[\widetilde{B}_j \mid \gamma(y) = r] \cdot \Pr[\gamma(y) = r],
\end{equation}
and let $\widetilde{B}_j(r)$ be the event that there exists a cut $A$ in $T_j$ such that all labeling variables of $A$ except $\kappa$ have placements \emph{at most} $r$.
By Property~\ref{pp_p} in Definition~\ref{def_frozen_tree}, in $T_j$ no vertex below the root labeled $y \neq \kappa$ has label $y$, thus event $\widetilde{B}_j(r)$ is equivalent to event $\widetilde{B}_j$ conditioned on $\gamma(y) = r$.
We shall use $\Phi_j$ to denote a lower bound of $\Pr[\widetilde{B}_j]$ and use $\phi_j(r)$ to denote a lower bound of $\Pr[\widetilde{B}_j(r)]$, then by Inequality~(\ref{ineq_B_d}), Equality~(\ref{eq_total_prob}), and Remark~\ref{hash_to_permutation}, in this section it remains to prove the second inequality in the following:
\begin{equation}\label{ineq_Phi_d}
    \Phi_d \ge \sum_{r \in [n]} \frac{\phi_d(r)}{n} \ge \lambda - o(1).
\end{equation}

We lower bound each term of the sum in Inequality~(\ref{ineq_Phi_d}):
\begin{lemma}\label{lem_subtree}
    For any integer $j \in [d]$ and any $r \in [n]$,
    \begin{equation*}
        \phi_j(r) \ge \left( \frac{r}{n} + (1 - \frac{r}{n}) \cdot \phi_{j-1}(r) \right)^{k-1} ,
    \end{equation*}
    where $\phi_0(r) = 0$.
\end{lemma}

\begin{proof}
    First of all, all variables that appear as labels different from $\kappa$ in tree $T_j$ take placements \emph{independently} and uniformly from $[n]$ by Remark~\ref{hash_to_permutation}, because there are at most $\tau$ such variables (Property~\ref{pp_l} in Definition~\ref{def_frozen_tree}).
    Let $v$ be the root of $T_j$ and $y \neq \kappa$ be its label.
    If $v$ has only one child and it is labeled by $\kappa$, then $\phi_j(r) = 1$ and the lemma holds.
    Otherwise, $v$ does not have a child labeled $\kappa$ by line~\ref{line_add_chi} of Algorithm~\ref{alg_construct_tree}.
    By Property~\ref{pp_k} in Definition~\ref{def_frozen_tree}, $v$ has at most $k-1$ children, and let them be $v_1, v_2, \dots, v_t$ with labels $y_1, y_2, \dots, y_t$ respectively, where $0 \le t \le k-1$.
    If $t = 0$ and $j > 0$ then $\phi_j(r) = 1$,
    thus the lemma holds.
    If $j = 0$ then $\phi_0(r) = 0$ since there is no cut, thus the lemma also holds.
    It remains to prove the case for $t \ge 1$ and $j \ge 1$.

    For event $\widetilde{B}_{j}(r)$ to happen, the event $Q_i \coloneqq (\gamma(y_i) \le r) \vee \widetilde{B}_{j-1}(r)$ must happen simultaneously for all $i \in [t]$.
    By Property~\ref{pp_p} in Definition~\ref{def_frozen_tree}, no label of vertex under $v_i$ is $y_i$, thus by the independence of all placements, the event $\gamma(y_i) \le r$ is independent of $\widetilde{B}_{j-1}(r)$.
    So we obtain
    \begin{equation}\label{ineq_Q_i}
        \Pr \left[ Q_i \right] \ge (\frac{r}{n} + (1 - \frac{r}{n}) \phi_{j-1}(r)),
    \end{equation}
    which immediately gives the lemma on $t = 1$.

    It remains to prove that for $t \ge 2$, the $Q_i$'s are positively correlated ($i \in [t]$), which is slightly different from the argument in \cite{DBLP:journals/jacm/PaturiPSZ05}.
    Let $V(T)$ be the set of all the labels except $\kappa$ appearing in the constructed $\tau$-frozen tree.
    Let $W \coloneqq W_r(\gamma)$ be the set of variables $z$ such that $z \in V(T)$ and $\gamma(z) \le r$, then each variable from $V(T)$ is in $W$ with probability $r / n$ independently, moreover each $Q_i$ only depends on $W$.
    Let $\widetilde{W}_i$ be the set of all subsets $W' \subseteq V(T)$ such that $W = W'$ implies $Q_i$.
    Since for any $W' \in \widetilde{W}_i$ it must be that the superset of $W'$ is also in $\widetilde{W}_i$, we have that $\widetilde{W}_i$ is a monotonically increasing family of subsets.
    Therefore by the FKG inequality (cf. Theorem~{6.3.2} in \cite{alon2016probabilistic}), we obtain:
    \begin{equation}\label{ineq_fkg}
        \Pr \left[ Q_1 \wedge Q_2 \right] = \Pr \left[ W \in \widetilde{W}_1 \cap \widetilde{W}_2 \right] \ge \Pr \left[ W \in \widetilde{W}_1 \right] \cdot \Pr \left[ W \in \widetilde{W}_2 \right] = \Pr \left[ Q_1 \right] \cdot \Pr \left[ Q_2 \right] .
    \end{equation}

    Observe that the intersection of two monotonically increasing families of subsets is also monotonically increasing, therefore by an induction on $t$ and Inequality~(\ref{ineq_fkg}) we have that
    \begin{equation*}
        \Pr \left[ \bigwedge_{i \in [t]} Q_i \right] \ge \prod_{i \in [t]} \Pr \left[ Q_i \right].
    \end{equation*}
    The lemma follows immediately from Inequality~(\ref{ineq_Q_i}) and $t \le k-1$.
\end{proof}

We borrow one analytical result from \cite{DBLP:journals/jacm/PaturiPSZ05}:
\begin{lemma}[cf. Lemma 8 in \cite{DBLP:journals/jacm/PaturiPSZ05}]\label{lem_analytical_lb}
    Given $y \in [0, 1]$.
    \begin{itemize}
        \item Let $f(x, y) \coloneqq (y + (1-y)x)^{k-1}$.
        \item Define the sequence $\{R_j(y)\}_{j \ge 0}$ by the recurrence $R_{j}(y) = f(R_{j-1}(y), y)$ and $R_0(y) = 0$.
        \item Define $R_j \coloneqq \int_{0}^{1} R_{j}(y) \, \mathrm{d}y$.
    \end{itemize}
    Then $R_d \ge \lambda - o(1)$ for $d = \lfloor \log_k \tau \rfloor = \Theta(\log\log n)$.
\end{lemma}

\begin{lemma}\label{lem_our_analytical_lb}
    $\Phi_d \ge R_d$.
\end{lemma}
\begin{proof}
    Firstly, we shall prove that $\phi_j(r) \ge R_j(r / n)$ holds for any $r \in [n]$ and any integer $j \in [0, d]$, by an induction on $j$.
    The case $j = 0$ is trivial by definition.
    Now suppose it holds for $j = i$, and we prove it for $j = i + 1$. We have:
    \begin{equation*}
        \phi_{i+1}(r) \ge \left( \frac{r}{n} + (1 - \frac{r}{n}) \cdot \phi_{i}(r) \right)^{k-1} \ge \left( \frac{r}{n} + (1 - \frac{r}{n}) \cdot R_i(\frac{r}{n}) \right)^{k-1} = f \left( R_i(\frac{r}{n}), \frac{r}{n} \right) = R_{i+1}(\frac{r}{n}) , 
    \end{equation*}
    where the first inequality is from Lemma~\ref{lem_subtree}, the second inequality is from the induction hypothesis, and the last two equalities are from Lemma~\ref{lem_analytical_lb}, completing the induction.

    Secondly, we show that for any $j \ge 0$, $R_j(y)$ is a non-decreasing function on $y \in [0,1]$, which is by an induction on $j$.
    Function $R_0(y)$ is a constant function.
    Suppose it holds for $j=i$, we shall prove it for $j = i +1$.
    Observe that the function $f(x, y)$ is non-decreasing on both $x$ and $y$ and has range $[0,1]$ when $x, y \in [0, 1]$,
    thus $R_{i+1}(y) = f(R_{i}(y), y)$ is non-decreasing on $y$ since $R_{i}(y)$ is non-decreasing on $y$ by the induction hypothesis.
    So the conclusion holds.

    Finally, putting everything together, we obtain:
    \begin{equation*}
        \Phi_d \ge \sum_{r \in [n]} \frac{\phi_d(r)}{n} \ge \sum_{r \in [n]} \frac{R_d \left(r / n \right)}{n} \ge \sum_{r \in [n]} \int_{\frac{r-1}{n}}^{\frac{r}{n}} R_{d}(y) \, \mathrm{d}y = R_d ,
    \end{equation*}
    where the third expression is called the \emph{right Riemann sum} of $\int_{0}^{1} R_{d}(y) \, \mathrm{d}y$ and gives an upper bound of the integral when the integrand is non-decreasing, completing the proof.
\end{proof}

Lemma~\ref{lem_analytical_lb} and Lemma~\ref{lem_our_analytical_lb} immediately give Inequality~(\ref{ineq_Phi_d}).
This completes the proof of Lemma~\ref{lem_new_bits} and hence of Theorem~\ref{thm_unique}.

\section{Partial Derandomization for the General Case}\label{sec_general}

In this section we prove Theorem~\ref{thm_main} by a simple reduction from the General case to the Unique case and applying Theorem~\ref{thm_unique}.

\begin{lemma}\label{lem_fix_liquid}
    For any $k$-CNF $F$ with $S \ge 1$ solutions, there exists a partial assignment $a$ such that $|V(a)| = \lceil \log S \rceil$ and $F_a$ has exactly one solution.
\end{lemma}
\begin{proof}
We shall explicitly construct a \emph{good} partial assignment $a$ using Algorithm~\ref{alg_construct_alpha}, such that $|V(a)| = \lceil \log S \rceil$ and $F_a$ has exactly one solution.
(Such partial assignment is an analysis tool only, which is \emph{not} known to the $k$-SAT algorithm.)

\begin{algorithm}[h]
\caption{$\textsf{Construct-}a(F)$}
\label{alg_construct_alpha}
\begin{algorithmic}[1]
\REQUIRE $k$-CNF $F$ with $S \ge 1$
\ENSURE partial assignment $a$
\STATE initialize assignment $a$ as an empty set
\WHILE {there exists a liquid variable $x \in V(F_a)$} {\label{line_iteration_begins}
    \STATE add literal $x$ to $a$ if $S(F_{a \cup \{x\}}) \le S(F_{a \cup \{\bar{x}\}})$, otherwise add literal $\bar{x}$ to $a$ \label{line_add_smaller_s}
}
\ENDWHILE\label{line_iteration_ends}
\WHILE{$|V(a)| < \lceil \log S \rceil$} {
\label{line_iteration2_begins}
    \STATE add literal $l$ corresponding to a variable in $V(F_a)$ to $a$ such that $F_{a \cup \{l\}}$ is satisfiable
}
\ENDWHILE\label{line_iteration2_ends}
\RETURN $a$
\end{algorithmic}
\end{algorithm}


    Let $a_i$ be the partial assignment at the end of the $i$-th iteration in the first loop (lines~\ref{line_iteration_begins}-\ref{line_iteration_ends}) in Algorithm~\ref{alg_construct_alpha}.
    By an induction, we shall prove that $1 \le S(F_{a_i}) \le S(F) / 2^i$ for all $i$.
    Then after at most $\lceil\log S\rceil$ iterations there must be no liquid variable, thus the remaining formula has only one solution.

    This trivially holds when $i = 0$.
    Since $x$ is liquid in $F_{a_i}$, both $F_{a_i \cup \{x\}}$ and $F_{a_i \cup \{\bar{x}\}}$ are satisfiable thus have at least one solution.
    Moreover, by $S(F_{a_i}) = S(F_{a_i \cup \{x\}}) + S(F_{a_i \cup \{\bar{x}\}})$ and line~\ref{line_add_smaller_s} we have that $S(F_{a_{i+1}}) \le S(F_{a_i}) / 2 \le S(F) / 2^{i + 1}$ by the induction hypothesis. Lines~\ref{line_iteration2_begins}-\ref{line_iteration2_ends} maintain satisfiability, so the lemma holds.
\end{proof}

The following inequality is widely used in information theory.
We include an one-line proof here for completeness:
\begin{lemma}\label{lem_bef}
    For any $\delta \in [0, 1]$, $\binom{n}{\delta n} \le 2^{\rho(\delta) n}$ where $\rho$ is the binary entropy function.
\end{lemma}
\begin{proof}
    Consider the binomial distribution with parameters $n$ and $\delta$:
    \begin{align*}
        1 = \sum_{i = 0}^n \binom{n}{i} {\delta}^i (1 - \delta)^{n-i}
        &\ge \binom{n}{\delta n} {\delta}^{\delta n} (1 - \delta)^{n - \delta n}
        = \binom{n}{\delta n} 2^{\delta n \log \delta + (1 - \delta) n \log(1-\delta)}
        = \binom{n}{\delta n} 2^{-\rho(\delta) n} ,
    \end{align*}
    giving the lemma.
\end{proof}

Our deterministic algorithm for General $k$-SAT is simple enough to be directly stated in the proof:
\begin{proof}[Proof of Theorem~\ref{thm_main}]
    Given $F$, we do not know $S$ in advance.
    The algorithm runs the following $n+1$ \emph{instances} in parallel (e.g., gives each instance an $\widetilde{O}(1)$ time slice on a sequential machine), and terminates whenever one of the instances terminates.
    For every integer $i \in [0, n]$, the $i$-th instance enumerates all possible combinations of $i$ variables from $V$, and for each combination enumerates all possible $2^i$ partial assignments $a$ on them, then tries to solve $F_a$ using the derandomized PPSZ from \S{\ref{sec_framework}} with cutoff time $2^{(1 - \lambda) (n - i) + o(n)}$, i.e., breaks the inner loop when it reaches the cutoff time and tries the next partial assignment.

    By Lemma~\ref{lem_fix_liquid}, in the $\lceil \log S \rceil$-th instance, there exist a combination of $\lceil \log S \rceil$ variables and a partial assignment $a$ on them such that $F_a$ has exactly one solution.
    By Theorem~\ref{thm_unique} and the previous paragraph, the instance returns a solution in time at most
    \begin{equation*}
        \binom{n}{\lceil \log S \rceil} \cdot 2^{\lceil \log S \rceil} \cdot 2^{(1 - \lambda) (n - \lceil \log S \rceil) + o(n)} \le \binom{n}{\lceil \log S \rceil} \cdot S^{\lambda} \cdot 2^{(1 - \lambda) n + o(n)}
        \le 2^{(1 - \lambda + \lambda \delta + \rho(\delta)) n + o(n)} ,
    \end{equation*}
    where the last inequality follows from setting $\delta \coloneqq (\log S) / n$ and applying Lemma~\ref{lem_bef}.
    The multiplicative overhead of this algorithm is at most $\widetilde{O}(n + 1) = \widetilde{O}(1)$, therefore Theorem~\ref{thm_main} follows.
\end{proof}

\section{Remarks} \label{sec_remark}

We presented a simple deterministic algorithm with upper bound as an increasing function of $S$, the number of solutions.
We know of at least two algorithms
(in fact, three, if including a random guessing for all variables),
whose upper bounds are decreasing functions of $S$: the PPZ algorithm \cite{DBLP:conf/focs/PaturiPZ97, DBLP:journals/jcss/CalabroIKP08} and Sch{\"{o}}ning's algorithm \cite{schoning1999probabilistic}, which, if properly derandomized, can be combined with ours to obtain a faster deterministic algorithm for General $k$-SAT.
For instance, simply running Sch{\"{o}}ning's algorithm and our algorithm concurrently gives a (randomized) algorithm faster than the current best deterministic $k$-SAT algorithm \cite{DBLP:conf/icalp/Liu18} when $k$ is large.
But the current derandomization of Sch{\"{o}}ning's algorithm \cite{dantsin2002deterministic, moser2011full} loses the benefit of running faster on formulae with more solutions.
Their original method using the covering code might not be able to overcome this drawback.


The ultimate problem is to fully derandomize PPSZ.
The method in \cite{DBLP:journals/jacm/PaturiPSZ05} for General $k$-SAT requires $\Omega(n)$-wise independence, thus not practical by enumeration.
We have not found a tighter upper bound for the number of guessed variables, which can be partly explained by the hard instance for PPSZ constructed in \cite{DBLP:conf/coco/SchederS17} with at least $(1 - \lambda + \theta) n$ guessed variables in expectation for some constant $\theta > 0$.

\bibliographystyle{alpha}
\bibliography{dPPSZ}

\end{document}